\documentclass{article}
\usepackage[utf8]{inputenc}
\usepackage[a4paper,margin=2cm]{geometry}
\usepackage{setspace}
\usepackage{amsmath}
\usepackage{amssymb}
\usepackage{amsthm}
\usepackage{times}
\usepackage{mathtools}
\usepackage{caption}
\usepackage{subcaption}
\usepackage{float}
\newtheorem{theorem}{Theorem}[section]

\newtheorem{lemma}[theorem]{Lemma}

\usepackage[backend=bibtex, natbib=true, bibstyle=authoryear, citestyle=authoryear, doi=false, sorting=nyt]{biblatex}

\DeclareMathOperator{\E}{E}

\usepackage{authblk}

\title{\Large Semi-Mechanistic Bayesian modeling of COVID-19 with Renewal Processes}
\author[1]{Samir Bhatt}
\author[1]{Neil Ferguson}
\author[2]{Seth Flaxman}
\author[2,{\textdagger}]{Axel Gandy}
\author[1]{Swapnil Mishra}
\author[2]{James A. Scott}

\affil[1]{\footnotesize MRC Centre for Global Infectious Disease Analysis, Imperial College, London, UK}
\affil[2]{\footnotesize Department of Mathematics, Imperial College, London, UK}
\affil[ ]{\footnotesize {\textdagger} \mbox{Corresponding author: a.gandy@imperial.ac.uk}}
\date{}

\begin{document}

\maketitle
\begin{abstract}
We propose a general Bayesian approach to modeling epidemics such as COVID-19. The approach grew out of specific analyses conducted during the pandemic, in particular an analysis concerning the effects of non-pharmaceutical interventions (NPIs) in reducing COVID-19 transmission in 11 European countries. The model parameterizes the  time varying reproduction  number $R_t$ through a regression framework in which covariates can e.g.\  be governmental interventions or changes in mobility patterns. This allows a joint fit across regions and partial pooling to share strength. This innovation was critical to our timely estimates of the impact of lockdown and other NPIs in the European epidemics, whose validity was borne out by the subsequent course of the epidemic. Our framework provides a fully generative model for latent infections and observations deriving from them, including deaths, cases, hospitalizations, ICU admissions and seroprevalence surveys. One issue surrounding our model's use during the COVID-19 pandemic is the confounded nature of NPIs and mobility. We use our framework to explore this issue. We have open sourced an R package epidemia implementing our approach in  Stan. Versions of the model are used by New York State, Tennessee and Scotland to estimate the current situation and make policy decisions.  
\end{abstract}

\section{Introduction} \label{sec:intro}

This article presents a general framework for semi-mechanistic Bayesian modeling of infectious diseases using renewal processes. The term semi-mechanistic relates to statistical estimation within some constrained mechanism. Variants of this general model have been used in specific analyses of Covid-19 \citep{Flaxman2020-cr, Flaxman2020, Vollmer2020-oj,Mellan2020-hi,Unwin2020-nc, NYS, Olney2020, Scotland, Mishra2020.11.24.20236661}, and continue to be used in ongoing work to make policy decisions. 
The present article motivates and discusses the key statistical and  epidemiological features of this framework, starting from a counting process setup. Various extensions of the basic model are considered, including a latent infection process. We discuss limitations and applications of the modelling framework to stimulate further research.

The model uses a flexible regression-based framework for parameterising transmission and ascertainment rates. This allows the fitting of multilevel models \citep{gelman_2006, hox_2010, kreft_2011} for several regions simultaneously. Such partial pooling of parameters has specific advantages in the context of infectious diseases. Suppose we wish to estimate the effect of NPIs \citep{Cowling2020, Flaxman2020} or mobility \citep{badr_2020, miller_2020} on transmission rates. Estimating these effects separately in different regions could lead to noisy estimates for at least two reasons. There is typically little high quality data at the early stages of an epidemic. Such data is generally correlated, reducing the information content that can be used to infer such an effect. In addition, NPIs often occur in quick succession and their effects are confounded \citep{hksar_2003, who_2003}. This is exacerbated by the random times between infections (the generation distribution) and between infections and observations, which smooths the observed data, making it more difficult to attribute changes in transmission rates to a particular NPI. Alternatively, one could pool the effect across all groups. This ignores group-level variations and can lead to poor predictive performance, in particular underestimating variance for previously unmodeled regions. One could augment such a model with group-level indicators, but this results in a large number of parameters, which are difficult to estimate and leads to overfitting with classical estimation techniques. Partial pooling provides a natural solution to this.

Sometimes the inferential goal is not to assess the effect of a covariate on outcomes, but rather to infer transmission rates over time. Previous studies have focused on estimating reproduction numbers from case data \citep{ferguson_2001, riley_2003, bettencourt_2008,fraser_2009, kelly_2010, Cori2013}, sometimes directly substituting observed case counts for the unknown number of infected individuals \citep{wallinga_2004}. However, the emergence of SARS-CoV-2 has highlighted shortcomings of methods that rely on just case data. Limited testing capacity at the early stages of the pandemic led to only a small proportion of infections being detected and reported \citep{Li_2020}. Those tested were typically more likely to have been hospitalised or were at higher risk of infection or death. In particular this proportion, referred to as the infection ascertainment rate (IAR) is country-specific and likely to have changed over time due to changes in testing policies and capacity. If unaccounted for, it will lead to biases in the inferred transmission rates.

This highlights the need for more flexible observational models, whereby more varied types of data can be incorporated, and their idiosyncrasies accounted for. Daily death data has been used in \citet{Flaxman2020} to recover reproduction numbers in the early stages of the SARS-CoV-2 pandemic, and has been seen as more reliable than case data. However, there have been clear variations in definitions and reporting across time and countries. It is therefore important to appropriately model noise within the the observational models. Our framework allows for multiple types of data including deaths, cases, hospitalizations, ICU admissions and the results of seroprevalence surveys. This improves robustness of inferred parameters to biases in any one type of data. 

The model uses discrete renewal processes to propagate infections within modeled populations. These have been used in a number of previous studies \citep{fraser_2007, Cori2013, nouvellet_2018, cauchemez_2008}, and are linked to other popular approaches to infectious disease modeling. \citet{champredon_2018} show that the renewal equation leads to identical dynamics as Erlang-Distributed Susceptible-Exposed-Infected-Recovered (SEIR) compartmental models, when a particular form is used for the generation distribution. A special case of this is the standard Susceptible-Infected-Recovered (SIR) model \citep{kermack_1927}. The approach is also connected to counting processes including the Hawkes process and the Bellman-Harris process \citep{bellman_1948, bellman_1952}. \citet{bellman_1948, Mishra2020-yt} derive the renewal equation as the expectation of an age-dependent branching process. Age-dependence allows for more realistic dynamics than age-insensitive processes, like the Galton-Watson process \citep{bartoszynski_1967, getz_2006}. More complex branching processes such as the Crump-Mode-Jagers branching process could also be considered. Hawkes processes are also related to renewal processes, with the expectation of the Hawkes intensity function resulting in the renewal equation \citep{Rizoiu_2017}.

We describe the general model in detail, and start by considering the bare-bones version in Section \ref{sec:overview}. The motivation for the model lies in continuous-time counting processes, and this connection is discussed in Section \ref{sec:motivation}. Sections \ref{sec:infections} and \ref{sec:observations} present the infection and observation processes in more detail, and consider important extensions of the basic model. Section \ref{sec:multilevel} considers how to use the framework for multilevel modeling. Section \ref{sec:reflections} compares our approach to standard time series models, and outlines the key challenges involved in modeling with our framework. Section \ref{sec:examples} considers the specific aspect of confounding and causality when estimating the effects of variables on transmission rates. Section \ref{sec:discussion} has a brief discussion.

\section{Model Overview}
\label{sec:overview}

We now formulate a basic version of the model for one homogeneous population. The same model can be used for multiple regions or groups jointly. Let $R_t > 0$ be the reproduction number at time $t>0$. This determines the rate at which infections grow. Infections $i_{v}, \ldots, i_0$ for some $v \leq 0$ are given a prior distribution. For $t > 0$, we let new infections $i_t$ be defined by
\begin{align}
i_t &= R_t \sum_{s < t} i_s g_{t-s}, \label{eq:renewal}
\end{align}
where the generation time,  the lag between infections, is given through a probability mass function $g$, i.e.\  $g_t\geq 0$ and $\sum_{t=1}^\infty g_t=1$.

Observations occur at certain times $t > 0$ . In general, there may be multiple types; case and death counts for example. Each such type is driven by its own time-varying ascertainment rate $\alpha_t > 0$. The mean of the observations at time $t$ is linked to past infections by 
\begin{align}
y_t&=\alpha_t\sum_{s \leq t} i_s\pi_{t-s}, \label{eq:yrecursion}
\end{align}
where $\pi$ is a distribution for the lag between an infection and when it gives rise to an observation. The sampling distribution of the observations with these means is typically nonegative and discrete, and may depend on auxiliary parameters. When multiple types are observed, we can superscript the quantities as $y^{(l)}_t, \alpha^{(l)}_t$ and $\pi^{(l)}$ and assign independent sampling distributions for each type.

Transmission rates $R_t$ and ascertainment rates $\alpha_t$ can be modeled flexibly using Bayesian regression models, and through sharing of parameters, are the means through which we tie together multiple regions or groups using multilevel modeling. One can, for example, model transmission rates as depending on a binary covariate for an NPI, say full lockdown. The coefficient for this can be \textit{partially pooled} between these groups. The effect is to share information between groups, while still permitting between group variation.

\section{Motivation from continuous time}
\label{sec:motivation}
Our model can be motivated from a continuous time perspective as follows. Infections give rise to additional infections in the future, referred to as offspring.  Letting $N^I(t)$ denote the number of infections occurring up to time $t$, defined by its intensity
\begin{equation} \label{eq:intensity}
    \lambda(t) = R(t) \int_{s < t} g(t-s)N^I(ds), \quad t>0,
\end{equation}
where $g$ is the density of a probability distribution on  $\mathbb{R}^+$ defining the time between infections, and
where  $\{R(t) : t  > 0\}$ is a non-negative stochastic process. 
The process can be initialised by assuming values for $N^I(t)$ for $t$ in the seeding period $[v,0]$. 

Equation \eqref{eq:intensity} is similar to the Hawkes intensity, however the \textit{memory kernel} $g$ is scaled by a time-specific factor $R(t)$. The integrand $g$ allows the intensity to increase due to previous infection events, while $R(t)$ tempers the intensity for other time-specific considerations. 
If $R(t') = R(t)$ for all $t'$ then Equation \eqref{eq:intensity} reduces to a Hawkes process. Under this assumption, since $g$ integrates to unity, the expected number of offspring is simply $R(t)$, and so this is the \textit{instantaneous reproduction number} or alternatively the \textit{branching factor} of the Hawkes process. The generation time, defined as the time from an infection to a secondary infection, is distributed according to $g$ and so $g$ is the \textit{generation distribution}.

Observations are precipitated by past infections; a given infection may lead to observation events in the future. Letting $N^Y(t)$ be the count of some observation type over time defined by the intensity
\begin{equation} \label{eq:yintensity}
    \lambda_y(t) = \alpha(t) \int_{s<t} \pi(t-s) N^I(ds),
\end{equation}
for $t>0$, where $\pi : \mathbb{R}^+ \to \mathbb{R}^+$ is a function and  $\{\alpha(t) : t \geq 0 \}$ is a non-negative stochastic process. This is similar to Equation \eqref{eq:intensity}, however the intensity increases due to past infections, rather than past observations. 

Consider the special case where $\pi$ is a probability density and where $\alpha(t') = \alpha(t)$ for all $t'$. The average number of observation events attributable to a single infection is then $\alpha(t)$, and so this is an \textit{instantaneous ascertainment rate}. $\pi$ is then interpreted as the distribution for the time from an infection to an observation, and therefore we call it the \textit{infection to observation} distribution.

\section{Infection Process} 
\label{sec:infections}

Starting from the continuous model, we now describe a discrete model, which results in the formulation of  Section \ref{sec:overview}. This discrete model is more amenable to inference. Let $ I_t$ be the number of new infections at time $t$; this is the equivalent of  $ N^I_t - N^I_{t-1}$ in the continuous model. As basic modelling block we use the following discrete version of \eqref{eq:intensity}:
\begin{equation} \label{eq:condexp}
\E[I_{t} \vert  R_{1:t}, I_{v:t-1}] = R_t L_t,
\end{equation}
where $L_t := \sum_{s < t} I_s g_{t-s}$ is the \textit{case load} or \textit{total infectiousness} by time $t > 0$.
Moreover,  letting $i_t:=\E[I_{t} \vert R_{1:t}, I_{v:0}]$ and taking the conditional expectation given reproduction numbers $R_{1:t}$ and seeded infections $I_{v:0}$ on both sides of \eqref{eq:condexp} gives
\begin{equation*}
i_t= R_t \E[L_t \vert R_{1:t}, I_{v:0}] = R_t\sum_{s<t} \E[I_{s} \vert R_{1:s}, I_{v:0}] g_{t-s}=R_t \sum_{s<t} i_s g_{t-s},
\end{equation*}
which is Equation \eqref{eq:renewal}. This is a discrete renewal equation, which can alternatively be interpreted as an $\text{AR}(t)$-process with known coefficients $g_{k}$. From this point of view, the basic model in Section \ref{sec:overview} uses $i_t$ as synonymous with actual infections. Since infections are simply a deterministic function of other parameters, there is no need to treat them as unknown latent parameters to sample. This can lead to lower sampling times and faster convergence.

\subsection{Modeling Latent Infections}
\label{sec:latentinf}

The model of Section \ref{sec:overview} can be extended by replacing each $i_t$ with the actual infections from the counting process $I_t$, and then assigning a prior to $I_t$. Although sampling can be slower, this has certain advantages. 
When past infection counts are low, significant variance in the offspring distribution can imply that new infections $I_t$ has high variance. This is not explicitly accounted for in the basic model. In addition, this approach cleanly separates infections and observations; the latter being modeled \textit{conditional} on actual infections. The sampling distribution can then focus on idiosyncrasies relating to the observation process.

We assign a prior to $I_t$ conditional on previous infections and current transmission $R_t$. The expected value for this is given by Equation \eqref{eq:condexp}. Appendix \ref{sec:offspring} shows that assuming the variance of the prior to be a constant proportion $d$ of this mean is equivalent to letting $d$ be the \textit{coefficient of dispersion} for the offspring distribution. $d > 1$ implies overdispersion, and can be used to account for super-spreading events, which has been shown to be an important aspect for modeling Covid-19. The parameter $d$ can be assigned a prior.

Any two parameter family can be used to match these first two moments. Letting this be continuous rather than discrete allows inference to proceed using Hamiltonian Monte Carlo, whereby new values for $I_t$ are proposed simultaneously with all other parameters. Possible candidates include log-normal, gamma and the Weibull distributions. If an explicit distribution for the offspring distribution is desired, one can show that assuming a Gamma distribution with rate $\lambda$ for this results in a Gamma distribution for $I_t$ with rate $\lambda$. The coefficient of dispersion is then simply $D = \lambda^{-1}$.

\subsection{Population Adjustments}
\label{sec:popadjust}

If $R_t$ remains above unity over time, infections grow exponentially without limit. In practice, infections should be bounded from above by $S_0$, the initial susceptible population. All else being equal, transmission rates are expected to fall as the susceptible population is diminished. 

Consider first the model using $I_t$, which was described in Section \ref{sec:latentinf}. Equation \ref{eq:condexp} can be replaced with
\begin{equation} \label{eq:condexprevised}
\E[I_{t} \vert R_{1:t}, I_{v:t-1}] = (S_0 - I_{t-1}) \left(1 - \exp\left(-\frac{R_{u,t}L_t}{S_0}\right)\right),
\end{equation}
where $R_{u,t}$ is an \textit{unadjusted} reproduction number, which does not account for the susceptible population.
This satisfies intuitive properties. As the \textit{unadjusted} expected infections $R_{u,t}L_t$ approaches infinity, the \textit{adjusted} expected value approaches the remaining susceptible population. The motivation for and derivation of Equation \eqref{eq:condexprevised} is provided in Appendix \ref{sec:popadjustapp}. In short, this is the solution to a continuous time model whose intensity is a simplification of Equation \eqref{eq:intensity}. We must also ensure that the distribution of $I_t$ cannot put positive mass above $S_0 - I_{t-1}$. A simple solution is to use truncated distributions. Of course, this adjusts the mean value from Equation \eqref{eq:condexprevised}, however this is unlikely to be significant unless the susceptible population is close to depletion.

In the basic model, one can apply the adjustment to $i_t$ by replacing $L_t$ in Equation \eqref{eq:condexprevised} with
\begin{equation}
    \E(L_t \vert R_{1:t}, I_{v:0}) = \sum_{s<t} i_s g_{t-s}.
\end{equation}

\section{Observations}
\label{sec:observations}

Observations are modeled in discrete time, analogous to how we treated infections in Section \ref{sec:infections}. Letting $\pi : \mathbb{N} \to \mathbb{R}^+$ and $Y_t := N^Y_t - N^Y_{t-1}$, the discrete analogue to Equation \eqref{eq:yintensity} is
\begin{equation} \label{eq:obscondmean}
    \E[Y_t \vert \alpha_t, I_{v:t}] = \alpha_t \sum_{s\leq t} I_s \pi_{t-s}.
\end{equation}
Taking the expected value of the above given seeded infections, transmission rates and the current ascertainment rate gives
\begin{equation} \label{eq:obscondmean2}
    \E[Y_t \vert \alpha_t, R_{1:t}, I_{v:0}] = \alpha_t \sum_{s\leq t} i_s \pi_{t-s},
\end{equation}
which is recognisable as Equation \eqref{eq:yrecursion}. Thus we have two possible expressions for the mean of $Y_t$, one given actual infections, and the other given expected infections $i_t$. The basic model of Section \ref{sec:overview} uses the latter, while the extension in Section \ref{sec:latentinf} uses the former.

We assume that $Y_t \sim \mathcal{F} \left(y_t, \phi \right)$, where $\mathcal{F}$ is a non-negative discrete family parameterised by its mean $y_t$ and potentially an auxiliary parameter $\phi$. This could be a Poisson distribution, where there is no auxiliary parameter. Using a quasi-Poisson or negative binomial instead allows for overdispersion. This can be useful to capture, for example, day-to-day variation in ascertainment rates when infection counts are low.  The mean $y_t$ can be taken to be either \eqref{eq:obscondmean} or \eqref{eq:obscondmean2}, the latter being used in the basic version of the model. Hidden in this formulation is the assumption that each $Y_t$ is conditionally independent given $y_t$. Using multiple observation series $Y^{(l)}_t$ can help to improve the model inferences and identifiability of certain parameters. We simply assume that each such series is conditionally independent given the underlying infection process.

\section{Multilevel Models}
\label{sec:multilevel}

Transmission rates can be modeled quite generally within the framework. If the aim is simply to estimate transmission in a single region over time, one approach could be to let $R_t = g^{-1}(\gamma_t)$, where $g$ is a link function and $\gamma_t$ is some autocorrelation process, for example a random walk. Suppose, however, that transmission is modeled in $M$ regions and the goal is to estimate the effect of a series of NPIs on transmission. Letting $R^{(m)}_{t}$ denote transmission in region $m$ at time $t$, we could let
\begin{equation}
    R^{(m)}_t = g^{-1}\left(  \beta^{(m)}_0 + \sum_{l=1}^p x^{(m)}_{t}\beta^{(m)}_{k}\right),
\end{equation}
where $x^{(m)}_{t}$ are binary encodings of NPIs, and  $\beta^{(m)}_0$ and $\beta_k^{(m)}$ are region-specific intercepts and effects respectively. The intercepts are used to allow regions to have their own baseline transmission rates. Collecting these group specific parameters into $\beta^{(m)}$, we can partially pool them by letting $\beta^{(m)} \sim \mathcal{N}\left(0, \Sigma \right)$, for each group $m$, and the assigning a prior to the covariance matrix $\Sigma$. This could be an inverse-Wishart prior, or alternatively, $\Sigma$ can be decomposed into variances and a correlation matrix, which are each given separate priors \citep{tokuda2011visualizing}.

One possible option for $g$ is is the log-link. This provides easily interpretable effect sizes; a one unit change in a covariate multiplies transmission by a constant factor. However, this can lead to prior mass on unreasonably high transmission rates. With this in mind, an alternative is to use a generalisation of the logit link for which
\begin{equation}
g^{-1}(x) = \frac{K}{1 + e^{-x}},
\end{equation}
and where $K$ is the maximum possible value for transmission rates. This serves a similar purpose to the carrying capacity in a logistic growth model.

The ascertainment rate $\alpha_t$ can also be modeled with similar considerations to the above. This flexibility is useful, particularly because these quantities are likely to change as an epidemic progresses. This has been clearly seen during the Covid-19 epidemic, where the infection ascertainment rate may have increased over time due to increased testing capacity and improved track and trace systems. Multilevel models can in theory also be specified through $\alpha_t$.

\section{Forecasting, epidemiological constants, and seeding}
\label{sec:reflections}

A key benefit of using a semi-mechanistic approach is that forecasts are constrained by plausible epidemiological mechanisms. For example, in the absence of any further interventions or behavioural changes, and looking at a medium term forecast of just incidence (daily new cases/infections), a traditional time series forecasting approach may predict a constant function based on observing broadly constant incidence, but semi-mechanistic approach would expect a monotonic decrease based on a constant rate of transmission and the effect of herd immunity. The performance of epidemiologically constrained models is generally good \citep{Carias_2019}; this is perhaps not surprising as examining the discrete renewal equation shows that these models correspond to autoregressive(n) filters with a convex combination of coefficients specified by the generation interval. However, similar to financial forecasting, the predictive capability of epidemic models are likely to be better interpreted as scenarios rather than actual predictions due to the rapidly adaptive landscape of policy.

A second benefit of epidemic models is to provide a plausible mechanism to explain (non causally) the changes observed in noisy data. For example, in estimating the effect of an intervention on observed death data, we need consider what that intervention affects i.e the rate of transmission or $R_t$. As we have described above, we can connect the rate of transmission to latent infections to deaths such that we have an epidemiologically motivated mechanism. While we can statistically estimate parameters for how the intervention affects $R_t$, certain important parameters will be entirely unidentifiable and need to be fixed as constants or with very tight priors. For example, to reliably estimate the number of infections, an infection fatality rate needs to be chosen. A failure to choose an appropriate infection fatality rate can result in a bimodal posterior where changes can either be attributed to herd immunity or to interventions. From a statistical perspective, is it difficult to disentanlge which mode of the posterior best represents reality, and hence it is sensible to first estimate a plausible infection fatality rate and then use this within the semi-mechanistic model. A second example is the onset of symptoms-to-death distribution. Given the lag between transmission, infection and deaths, the effect of an intervention is dependent on the onset of symptoms-to-death distribution. 

Infection seeding is a fundamentally challenging aspect of epidemic modelling. Estimating the initial effect of seeding is crucial to understanding a baseline rate of growth ($R_0$) from which behaviours and interventions can modify. This seeding is heavily confounded by importation and under ascertainment. Both these factors can influence estimates of the initial growth rates, and this in turn can affect the impact of changes in transmission as time progresses. We have proposed heuristic approaches to mitigate issues with early seeing, but principled statistical approaches need to be developed. In particular, Bayesian pair plots show strong correlation between seeding parameters and $R_0$, which can potentially lead again to a bimodal posterior where initial growth dynamics can be explained through $R_0$ or via initial seed infections.

\section{Confounding and Causality: Estimating the Effect of Interventions}
\label{sec:examples}

Section \ref{sec:multilevel} showed that changes in transmission rates over time can be explained by parameterising these rates in terms of predictor variables, such as NPIs and mobility. Clarifying the effects of interventions on transmission is important, if only because of their clear economic and human consequences. This is however a significant challenge because the effects are potentially confounded with unobserved behavioural changes, and they are hard to identify. Identification is difficult because interventions are highly correlated when they occur in quick succession. Moreover, the random time between an infection and its recording as a case or death leads to observations being less informative about the effect of any particular intervention.

\citet{Flaxman2020} estimated the effectiveness of NPIs across 11 European countries, and used partial pooling of effect sizes to address the identification problem. At that time, little data existed other than information on deaths and the timing of interventions. NPIs, which were coded as a binary set of mandatory government measures (e.g.~school closures, ban on public events, lockdown), could not fully explain the patterns seen in some countries (e.g.~Sweden), and especially at the subnational level.  Mobility data became available in April and was used to model the epidemic in Italy, Brazil and the USA ~\citep{Vollmer2020-oj,Mellan2020-hi,Unwin2020-nc}. Such data is useful as it may help account for behavioral changes that confound the effects of NPIs. However since mobility affects transmission, is linked to the introduction of NPIs and potentially also to voluntary behavioural measures, we expect it to be a confounder.

Section \ref{sec:twostage} extends the model in \citet{Flaxman2020} to investigate further this issue of confounding, and models both NPIs and mobility jointly. This is in keeping with standard practice in regression/ANOVA: expanding a model to take into account more explanatory variables. Nonetheless, NPIs may partially affect transmission \textit{via a path through mobility}. A joint model of mobility and NPIs does not account for this. Therefore, in Section \ref{sec:causal} we take a first step in assessing causal considerations through a simple mediation analysis. We begin however by exploring the relationship between interventions and mobility.

\subsection{Interventions and Mobility}
\label{sec:intandmob}
Regressing average mobility on NPIs in a Bayesian linear model (no intercept or partial pooling) we find a correlation of over $85\%$ with a mean absolute error of $0.1\%$. Given mobility generally ranges from -1 to 1, this is a good overall fit. Figure \ref{fig:effectsize}a shows that these fits visually correspond well with changes in average mobility. One could conjecture that mobility and NPIs are lagged, but lagging NPI dates either forwards or backwards in time does not result in a better fitting model (see Figure \ref{fig:effectsize}b). Indeed Figure 1b does support a hypothesis that the timing of NPIs and changes in mobility are strongly linked. The coefficient sizes from this regression are entirely consistent with \citet{Flaxman2020} finding  that the NPI with the largest effect size is lockdown (see Figure \ref{fig:effectsize}c). This simple analysis does not model transmission, but does provide strong evidence that mobility and NPIs do not provide conflicting narratives. We note, to perform this regression as fairly as possible we used a hierarchical shrinkage prior \citep{Piironen2017-qs} that performs both shrinkage and variable selection.

 \begin{figure}[H]
	\centering
	\begin{subfigure}{.33\textwidth}
		\centering
		 \includegraphics[width = \textwidth]{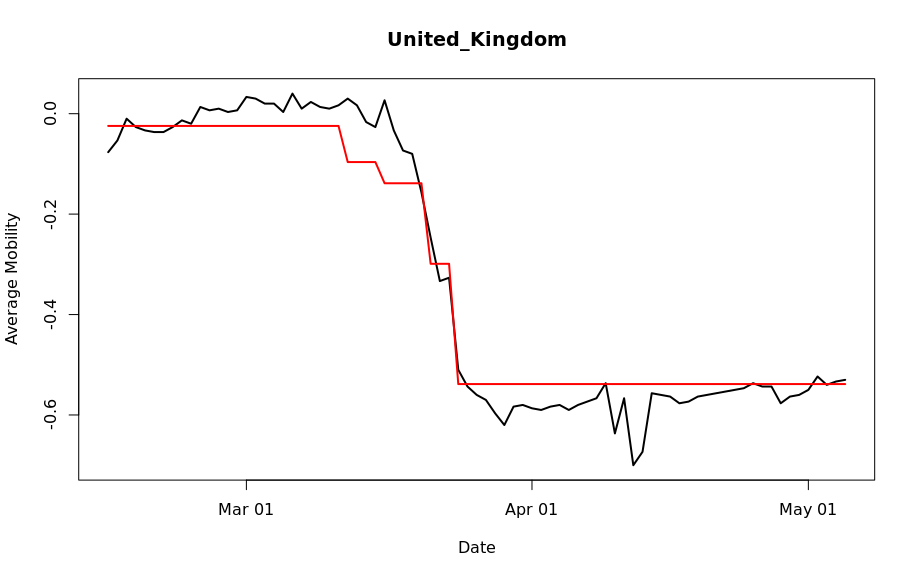}
		\caption{~}
	\end{subfigure}%
	\begin{subfigure}{.33\textwidth}
		\centering
		\includegraphics[width = \textwidth]{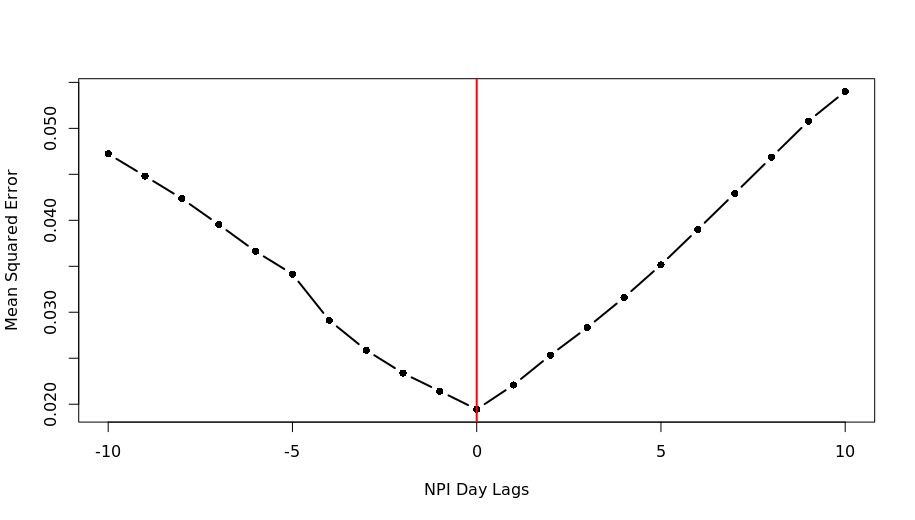}
		\caption{~} 
	\end{subfigure}%
\begin{subfigure}{.33\textwidth}
	\centering
	\includegraphics[width = \textwidth]{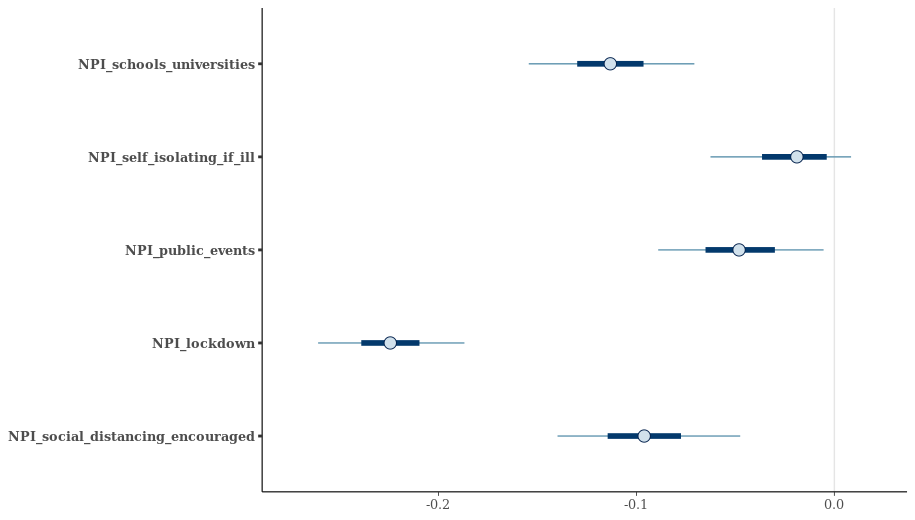}
	\caption{~} 
\end{subfigure}%
\caption{Simple regression of mobility against NPIs. Figure (a) shows the fit for the United Kingdom with mobility in black and the fit for NPIs in red. (b) Shows the effect on mean absolute error from lagging the NPIs. (c) Shows the coefficient effect sizes from the regression. Y axis are NPIs and X axis the regression effect sizes}
\label{fig:effectsize}
\end{figure}

\subsection{Controlling for Mobility}
\label{sec:twostage}

Section \ref{sec:intandmob} found a large correlation between interventions and mobility, demonstrating that mobility is a potential confounder. Here we control for this by jointly modeling NPIs and mobility. This is done using the same 11 European countries, sets of NPIs and death data as used in \citet{Flaxman2020}. 

A two-stage approach \citep{Haug2020-fq} is used, whereby $R_t$ is first estimated using a daily random walk. In the second stage this is regressed on NPIs and mobility. The random walk can in theory select any arbitrary function of $R_t$ that best describes the data without any prior information about which interventions happened when or how well they worked. Given these estimates of $R_t$ for all 11 European countries, we run a simple partial pooling model to see if interventions and/or mobility can reproduce the trends in $R_t$. The model used is a linear regression with country level intercepts (to account for variation in $R_0$), and both joint and country specific effect sizes for interventions/mobility. As with the earlier analysis we use a hierarchical shrinkage prior on the coefficients \citep{Piironen2017-qs}.

Three variations of the model are considered: NPIs only (NPI\_only), mobility only (Mobility\_only), and both NPIs and mobility (NPI+Mobility). MCMC convergence diagnostics in all cases did not indicate problems. We found the best fitting model to be NPI+Mobility. Relative to this model the expected log posterior difference ($\pm$ standard error) in WAIC of the model with only NPIs is $-5.2\pm 4$, ad $-565.6\pm 49.2$ with only mobility. Therefore, in fits to the estimated $R_t$, the model with mobility alone is substantially worse than the models with NPIs.
Controlling for mobility does not appear to significantly change the estimated effects of NPIs. As in \citet{Flaxman2020}, the largest effect size is attributed to lockdown, as seen in Figure~\ref{fig:two-stage-effectsize}. This is true with and without the inclusion of the mobility variable. This analysis could be improved using Bayesian leave-one-out cross-validation \citep{Vehtari2017-fy} to account for the time series nature of the data.

\begin{figure}[H]
	\centering
	\begin{subfigure}{.33\textwidth}
		\centering
		 \includegraphics[width = \textwidth]{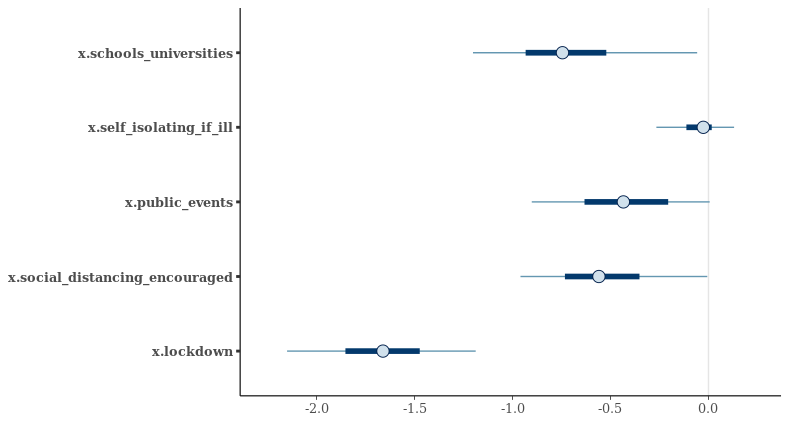}
		\caption{~}
	\end{subfigure}%
	\begin{subfigure}{.33\textwidth}
		\centering
		\includegraphics[width = \textwidth]{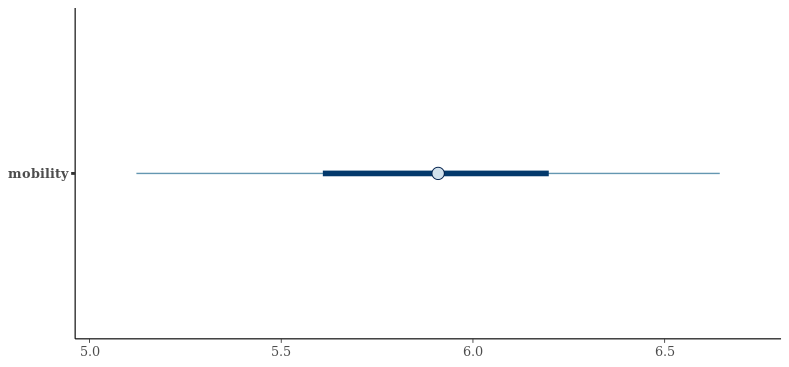}
		\caption{~} 
	\end{subfigure}%
\begin{subfigure}{.33\textwidth}
	\centering
	\includegraphics[width = \textwidth]{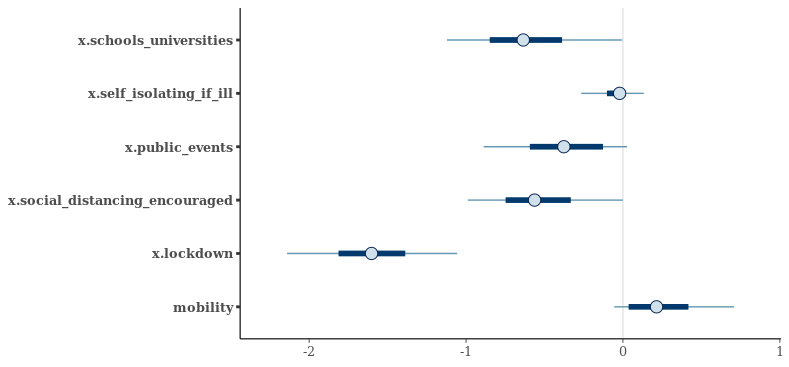}
	\caption{~} 
\end{subfigure}%
\caption{Regression of NPIs and/or mobility against nonparametric $R_t$. Figure (a) shows the fit for  NPIs only. (b) for  mobility only, (c) NPIs and mobility. Mobility only was not \emph{significantly} preferred by WAIC.  Y axis are covariates and X axis the regression effect sizes}
\label{fig:two-stage-effectsize}
\end{figure}

An advantage of the two-stage approach is that it is scalable to a large number of regions. $R_t$ can be estimated in each region in parallel using separate models. Partial pooling can still be leveraged to estimate effects in the second stage. Once $R_t$ has been estimated, any number of interesting statistical analyses can be conducted. Nonetheless the estimated $R_t$ is not entirely non-parametric; it is clearly influenced by its random walk form in the first stage. This analysis could be extended by considering a range of alternative priors for $R_t$. More importantly however, this approach has not considered causal relationships between NPIs and mobility. This is the focus of the next example. 

\subsection{Causal Mediation}
\label{sec:causal}

The effect of interventions in the previous analysis holds mobility constant. However, we intuitively expect that part of these effects occur indirectly through changeing mobility. We can hypothesise that changes in mobility are both an effect of NPIs and a cause of reductions in transmission. Causal mediation analysis provides a means to disentangle the total effect of a variable into a direct and indirect effect. The indirect effect occurs via some mediator, which in this case is hypothesised to be mobility. Further information about causal mediation can be found here \citep{Pearl2009-dw,Pearl2012-hb}.

Only lockdown is considered here because performing causal mediation with all NPIs is challenging and lockdown is consistently the NPI with the largest effect size in Section \ref{sec:intandmob}, Section \ref{sec:twostage} and in \citet{Flaxman2020}. Briefly, to perform causal mediation we 
consider two transmission models
\begin{align}
R_{t}^{(m)} &= \tilde{R}^1_{m} \exp \left(\left(\beta^1_1 + \beta^1_{1,m}\right) L_{t,m} + \epsilon^1_{t,m} \right), \label{model:Rtm_causal} \\ 
R_{t}^{(m)} &= \tilde{R}^2_{m} \exp \left(\left(\beta^2_1 + \beta^2_{1,m}\right) L_{t,m} + \left(\beta^2_2 + \beta^2_{2,m}\right) M_{t,m} + \epsilon^2_{t,m}\right), \label{model:Rtm_causal2}
\end{align}
where $L_{t,m}$ is  a binary indicator for lockdown and $M_{t,m}$ is mobility in country $m$ respectively. $\tilde{R}^i_{m}$ and $\epsilon^i_{t,m}$ are country specific parameters modeling baseline transmission and a weekly random walk respectively. All other aspects of both models are the same as in \citep{Flaxman2020}. Model \eqref{model:Rtm_causal} includes effects for lockdown, while \eqref{model:Rtm_causal2} additionally considers mobility. $\beta^1_1$ is the total effect for lockdown, while $\beta^2_1$ is the partial effect when controlling for mobility. The mediated effect is therefore $\beta^1_1 - \beta^2_1$. This quantifies the effect of lockdown \textit{via the path through mobility}. We find this mediated effect reduces $R_t$ by 18.3\% with a 95\% credible interval of $[12.2\%, 44.4\%]$. The posterior probability of the effect being greater than $0$ is 89.6\% . Individual coefficients are shown in Figure \ref{fig:mediation}.
\begin{figure}[H]
    \centering
    \begin{subfigure}{.45\textwidth}
        \centering
        \includegraphics[width = \textwidth]{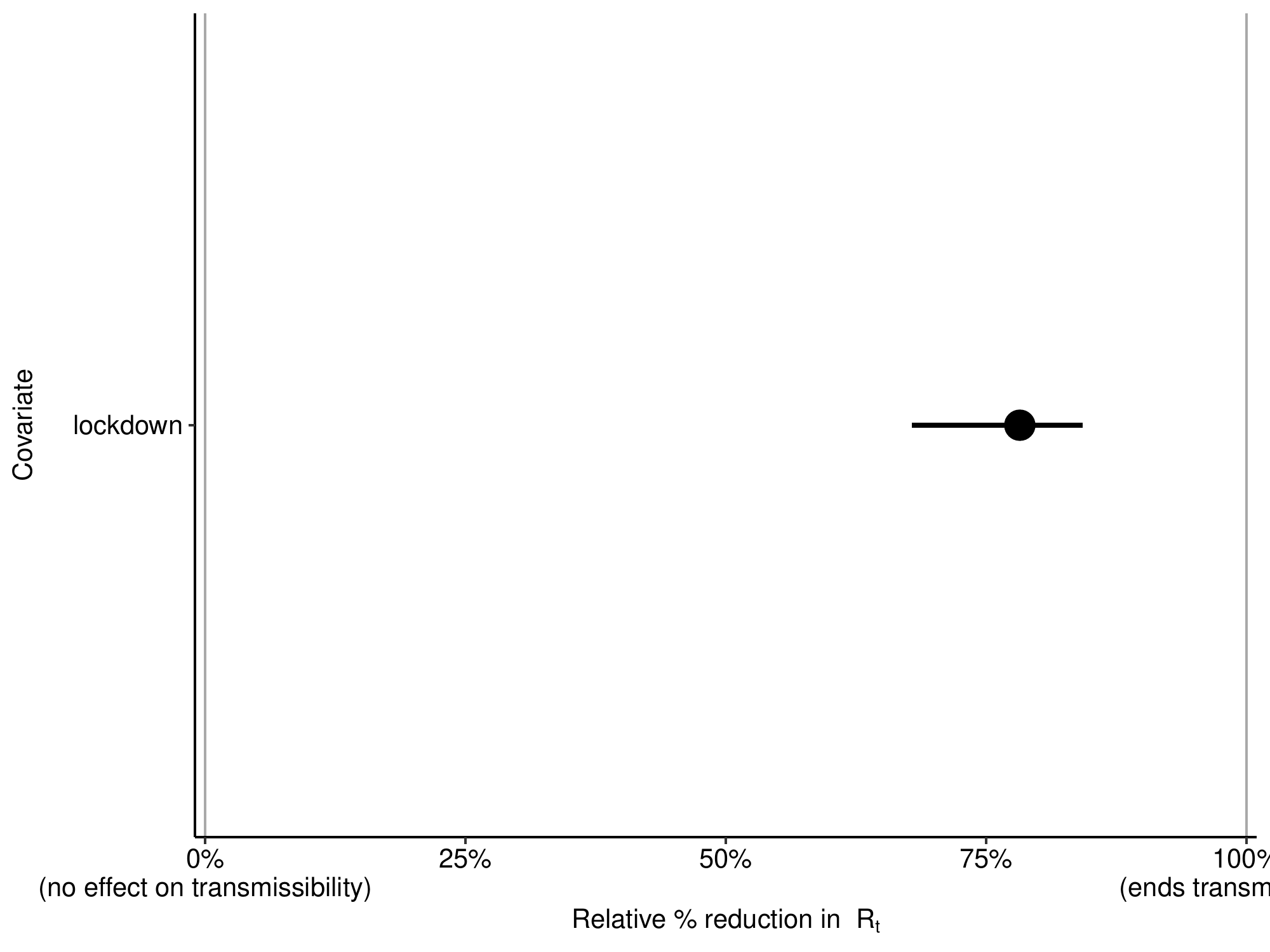}
        \caption{~}
    \end{subfigure}%
    \begin{subfigure}{.45\textwidth}
        \centering
        \includegraphics[width = \textwidth]{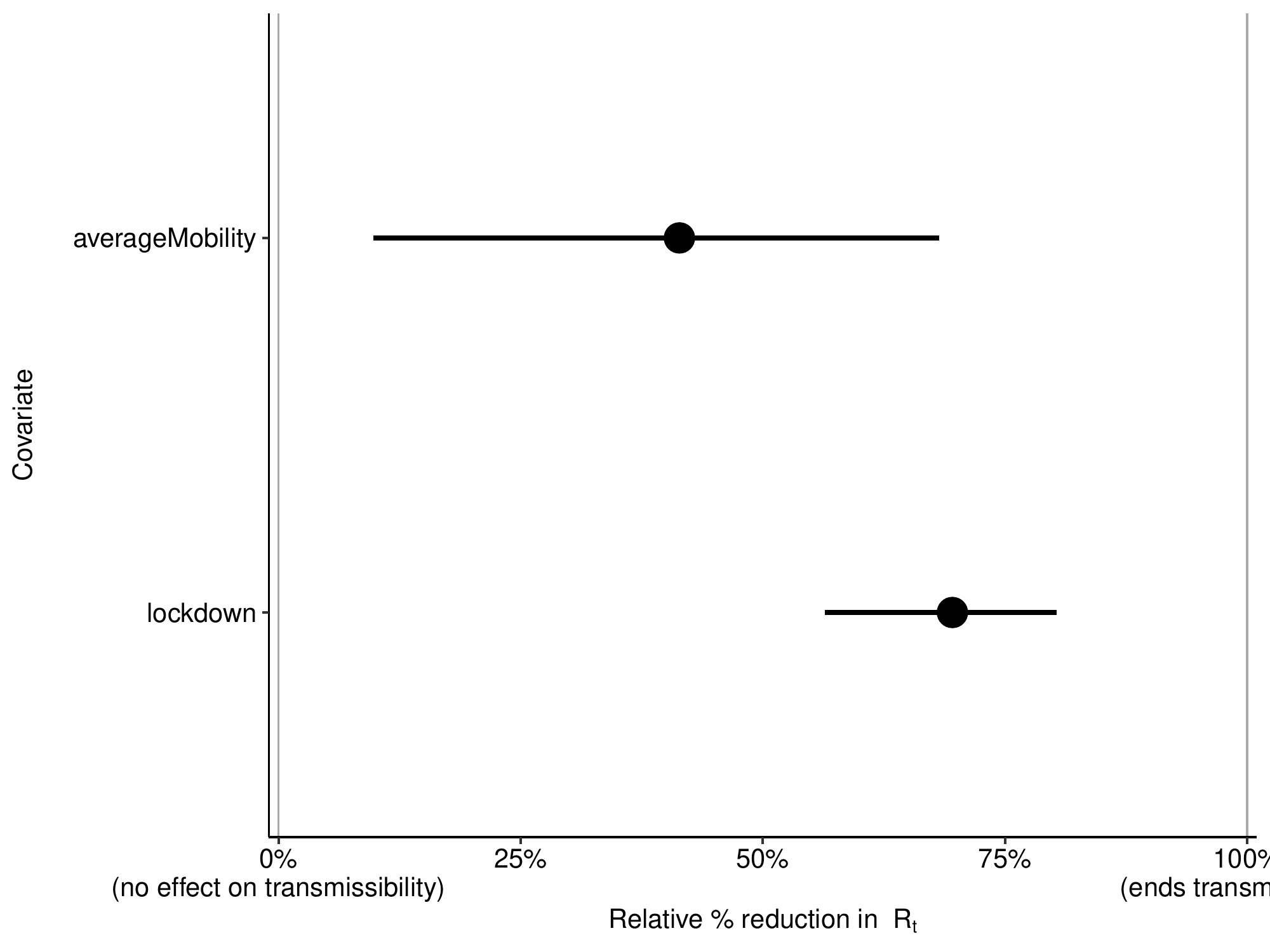}
        \caption{~} 
    \end{subfigure}%
    \caption{Mediation analysis (a) Lockdown, (b) Lockdown and Mobility}
    \label{fig:mediation}
\end{figure}

These mediation results suggest a causal link between lockdown and mobility that eventually leads to reduced transmission rates. They also suggest that the mediated effect is far less than the total effect of lockdown, suggesting lockdown will have other causal pathways. Of course, mobility is also mediated through other pathways, and a principled causal analysis is out of the scope of this article. The exclusion of other NPIs may introduce omitted variable bias. Nonetheless, this simple analysis with lockdown adds support to causal considerations.

\section{Discussion}
\label{sec:discussion}

This article has discussed a class of statistical models for epidemics such as Covid-19 which are able to capture key epidemiological mechanisms. The model has appeared in various forms for specific analyses during the Covid-19 crisis and, at the time of writing, continues to be used to inform public policy. By presenting it in a general form and discussing key modelling difficulties we hope to stimulate discussion around it. One key difficulty within the framework is dealing with confounded variables, particularly those used to explain changes in transmission during the early stages of an epidemic. The analyses in Section \ref{sec:examples} make a first step in dealing with these, and support the central finding of \citet{Flaxman2020}: that lockdown and other NPIs together served to control the first wave of the epidemic in 11 European countries. 

A number of model enhancements have not been discussed here. These include explicitly accounting for importations, and allowing for uncertainty in the generation and infection to observation distributions. The model can be fit using Stan \citep{stan_2020}, but the adaptive Hamiltonian Monte Carlo used often has difficulty converging when latent infections are modeled directly, or when multiple regions are jointly modeled. We conjecture that convergence may be improved by carefully choosing initial parameters for the sampler. Future research could explore whether alternative samplers can be developed to fit these models more pragmatically.

\printbibliography

\section{Appendix}

\subsection{Offspring Dispersion}
\label{sec:offspring}

Define the offspring distribution of any given infection to be the distribution of the random number of offspring attributable to that infection. We show that assuming the variance of these distributions are a constant proportion of the mean implies, under suitable independence assumptions, the same result for new infections $I_t$ for all time points.

Assume some ordering over infections at each period, and let $O_t^{(i)}$ denote the number of offspring of the $i$\textsuperscript{th} infection at time $t$. This can be decomposed as
\begin{equation} \label{eq:decomp}
    O_{t}^{(i)} = \sum_{s = t + 1}^{\infty} O_{ts}^{(i)},
\end{equation}
where $O_{ts}^{(i)}$ are the number of offspring of $i$ birthed at time $s$. The branching process behind Equation \eqref{eq:condexp} implies that $O_{ts}^{(i)}$ has mean $R_s g_{s-t}$. Assume that $\{O^{(i)}_{ts} : s \geq t\}$ are mutually independent and have variance which is a fixed proportion $d$ of the mean. By Equation \eqref{eq:decomp}, this implies the same variance relationship for $O_{t}^{(i)}$. In particular, if $R_s = R_t$ for $s > t$ then $O_{t}^{(i)}$ has mean $R_t$ and variance $d R_t$. New infections at time $t$ can be expressed as
\begin{equation} \label{eq:itdecomp}
    I_t = \sum_{s=1}^{t-1} \sum_{i=1}^{I_s} O_{st}^{(i)}.
\end{equation}
Assume that all $O_{st}^{(i)}$ appearing in Equation \eqref{eq:itdecomp} are mutually independent conditional on everything occurring up to time $t-1$, the result clearly follows by taking the variance of both sides of Equation \ref{eq:itdecomp} given $R_t$ and $I_{v:t-1}$.

\subsection{Population Adjustment}
\label{sec:popadjustapp}

Here we motivate Equation \eqref{eq:condexprevised}, which is used to adjust transmission rates for the size of the infectable population. The most obvious starting point for such an adjustment would be to let 
\begin{equation} \label{eq:condexpadj}
\E[I_{t} \vert R_t, I_{v:t-1}] = \left(\frac{S_0 - I_{t-1}}{S_0} \right) R_{u,t} L_t,
\end{equation}
where $R_{u_t}$ is defined as in Section \ref{sec:popadjust}. This is similar in form to a \textit{discrete logistic growth model}. Such models are well known as examples of simple models that exhibit chaotic dynamics \citep{May1976}. In particular, it is possible that the expected value on the left hand side exceeds the remaining susceptible population. Intuitively, this issue occurs because multiple infections can occur simultaneously in the discrete model. We therefore propose solving this by using a population adjustment motivated by the solution to a continuous time model whose intensity is a simplification of Equation \eqref{eq:intensity}.

Suppose we observe $I_{v:t-1}$ and current transmission $R_t$. We evolve infections from time $t-1$ to $t$ continuously, and hence avoid overshooting. Define a continuous time counting $\tilde{I}(s)$ process starting at time $t-1$ by the intensity
\begin{equation} \label{eq:modifiedintensity}
    \tilde{\lambda}(s) = \left(\frac{S_0 - \tilde{I}(s)}{S_0} \right) R_{u,t} L_t,
\end{equation}
for $s \geq t-1$, and with initial condition $\tilde{I}(t-1) = I_{t-1}$. Supplementary \ref{sec:popadjustproof} shows that 
\begin{equation} \label{eq:popadjustsupp}
    \E[\tilde{I}(t)] = I_{t-1} + (S_0 - I_{t-1}) \left(1 - \exp\left(-\frac{R_{u,t}L_t}{S_0}\right)\right),
\end{equation}
which is the motivation for Equation \eqref{eq:condexprevised}.

\newpage

\section{Online Supplement}

\subsection{Proof of Equation \eqref{eq:popadjustsupp}}
\label{sec:popadjustproof}

Without loss of generality, we prove the result for time $t = 1$. The argument remains the same for all $t > 1$.

From \cite[Lemma 5.5]{Thompson1985}, we have
$$
\E[\tilde{I}(s)] = \tilde{I}(0) + \int^s_0 \E[\tilde{\lambda}(l)] dl \text{ for } s \geq 0. 
$$
The following lemma derives an expression for the the expected intensity on the right hand side.

\begin{lemma} \label{lemma:intensity}
The expected intensity takes the form 
\begin{equation*}
    E[\tilde{\lambda}(s)] = \tilde{\lambda}(0)\exp\left( - \frac{R_{u,1}L_1}{S_0}s\right),
\end{equation*}
for all $s \geq 0$.
\end{lemma}

\begin{proof}[Proof of Lemma \ref{lemma:intensity}]
Fix $s \geq 0$, some small $\Delta > 0$ and let $h(s) := \E[\tilde{\lambda}(s)]$. We have from Equation \eqref{eq:modifiedintensity} that
\begin{equation} \label{eq:htdelta}
    h(s+\Delta) = \left(\frac{S_0 - \E[\tilde{I}(s+ \Delta)]}{S_0} \right) R_{u,1} L_1.
\end{equation}
We can write
\begin{equation*}
    \E[\tilde{I}(s + \Delta) \vert \tilde{\lambda}(s)] =  \E[\tilde{I}(s) \vert \tilde{\lambda}(s)] + \tilde{\lambda}(s) \Delta + \mathcal{O}(\Delta),
\end{equation*}
and taking expectations on both sides,
\begin{equation*}
\E[\tilde{I}(s+ \Delta)] = \E[\tilde{I}(s)] + h(s)\Delta + \mathcal{O}(\Delta).
\end{equation*}
Substituting this into \eqref{eq:htdelta} and rearranging gives
\begin{equation*}
    \begin{split}
        h(s+\Delta)  &= \left(\frac{S_0 - \E[\tilde{I}(s+ \Delta)]}{S_0} \right) R_{u,1} L_1 - \frac{R_{u,1}L_1}{S_0}\left(h(s) \Delta + \mathcal{O}(\Delta)\right), \\
        & = h(s) - \frac{R_{u,1} L_1}{S_0}\left(h(s) \Delta + \mathcal{O}(\Delta)\right).
    \end{split}
\end{equation*}
Rearranging gives
\begin{equation*}
    \frac{h(s+\Delta) - h(s)}{\Delta} = -\frac{R_{u,1}L_1}{S_0}\left(h(s)  + \frac{\mathcal{O}(\Delta)}{\Delta}\right).
\end{equation*}
Taking the limit as $\Delta \to 0$ and rearranging gives the differential equation
\begin{equation*}
    \frac{h'(s)}{h(s)} = -\frac{R_{u,1}L_1}{S_0}.
\end{equation*}
Integrating both sides gives
\begin{equation*}
    \log(h(s)) = -\frac{R_{u,1}L_1}{S_0}s + C.
\end{equation*}
Using that $h(0) = \tilde{\lambda}(0)$ gives the constant $C = \log(\tilde{\lambda}(0))$. Plugging in yields the required result.
\end{proof}

Hence, 
\begin{equation*}
    \begin{split}
        \E[\tilde{I}(s)] & = I_0 + \tilde{\lambda}(0) \int_0^s \exp\left( - \frac{R_{u,1}L_1}{S_0}l\right) dl\\
        & = I_0 + \tilde{\lambda}(0)\frac{S_0}{R_{u,1}L_1} \left( 1 - \exp\left( - \frac{R_{u,1}L_1}{S_0}s\right)\right) \\
        &= I_0 + (S_0 - \tilde{I}(s))\left( 1 - \exp\left( - \frac{R_{u,1}L_1}{S_0}s\right)\right).
    \end{split}
\end{equation*}
Letting $s = 1$ gives the required result.

\end{document}